\documentclass[sigconf,nonacm]{acmart}

\AtBeginDocument{%
  \providecommand\BibTeX{{%
    \normalfont B\kern-0.5em{\scshape i\kern-0.25em b}\kern-0.8em\TeX}}}

\usepackage{subcaption}
\usepackage{makecell}
\usepackage{algpseudocode}
\usepackage{algorithm}
\usepackage{placeins}
\usepackage{amsfonts}
\usepackage{booktabs}
\usepackage{enumitem}
\usepackage{colortbl}
\usepackage[toc,page]{appendix}
\usepackage{multirow}
\newfloat{procedure}{tbp}{lop}
\floatname{procedure}{Procedure}

\usepackage{times}
\usepackage{soul}
\usepackage{url}
\usepackage[utf8]{inputenc}
\usepackage{graphicx}
\usepackage{multicol}
\usepackage{bbm}

\usepackage{amsmath,amssymb,amsthm}
\newtheoremstyle{break}%
  {}{}%
  {\itshape}{}%
  {\bfseries}{}%
  {\newline}%
  {}%
\theoremstyle{break}

\usepackage{xcolor}
\usepackage{geometry}
\usepackage{hyperref}
\hypersetup{
    colorlinks = false,
    pdfborder = {0 0 0},
}
\usepackage{tikz}
\newcommand{\boxedref}[1]{%
    \tikz[baseline=(char.base)]{%
        \node[shape=rectangle,draw=red,inner sep=2pt] (char) {\hyperref[#1]{\ref*{#1}}};%
    }%
}

\newtheorem{theorem}{Theorem}[section]

\newtheorem{lemma}[theorem]{Lemma}
\newtheorem{assumption}{Assumption}
\newtheorem{property}{Property}

\usepackage{arydshln}

\DeclareRobustCommand{\okina}{%
  \raisebox{\dimexpr\fontcharht\font`A-\height}{%
    \scalebox{0.8}{`}%
  }%
}

\begin{document}

\title{A Statistically Reliable Optimization Framework for Bandit Experiments in Scientific Discovery}

\author{Tong Li}
\affiliation{
  \institution{University of Toronto}
  \country{Canada}}
\email{tongli@cs.toronto.edu}

\author{Travis Mandel}
\affiliation{
  \institution{University of Hawai{\okina}i at Hilo}
  \country{United States}}
\email{tmandel@hawaii.edu}

\author{Goldie Phillips}
\affiliation{
  \institution{Independent Researcher}
  \country{Trinidad and Tobago}
  \email{phillipstgoldie@gmail.com}
}

\author{Anna Rafferty}
\affiliation{
  \institution{Carleton College}
  \country{United States}}
\email{arafferty@carleton.edu}

\author{Eric M. Schwartz}
\affiliation{
  \institution{University of Michigan}
  \country{United States}
  \email{ericmsch@umich.edu}
}

\author{Dehan Kong}
\affiliation{
  \institution{University of Toronto}
  \country{Canada}}
\email{dehan.kong@utoronto.ca}

\author{Joseph J.~Williams}
\affiliation{
  \institution{University of Toronto}
  \country{Canada}}
\email{williams@cs.toronto.edu}

\begin{abstract}

Scientific experimentation is largely driven by statistical hypothesis testing to determine significant differences in interventions. Traditionally, experimenters allocate samples uniformly between each intervention. However, such an approach may lead to suboptimal outcomes - multi-armed bandits (MABs) addresses this problem by allocating samples adaptively to maximize outcomes. Yet, two challenges have hindered the use of MABs in scientific domains. First, common hypothesis tests (e.g., $t$-tests) become invalid under adaptive sampling without correction, leading to inflated type~I and type~II errors. This is an understudied problem, and prior solutions suffer from issues such as low statistical power which prevent adoption in many practical settings. Second, practitioners must explicitly balance cumulative reward with statistical efficiency, yet no general methodology exists to quantify this trade-off across algorithms. In this paper, we study assumption modification and critical region correction approaches for hypothesis testing that enable common tests to be applied to adaptively collected data. We provide heuristic justification for its power efficiency and show in simulation that it achieves higher power than existing approaches. Further, we derive a theoretically and practically motivated objective function for adaptive experiment evaluation, which we integrate into a unified experimental framework. Our framework asks experimenters to specify an \textit{experiment extension cost} for their problem, and based on that enables our proposed optimization procedure to select the bandit algorithm that best balances reward and power in their setting. 
We show that our approach enables practitioners to improve outcomes with only slightly more steps than uniform randomization, while retaining statistical validity.

\end{abstract}

\keywords{adaptive experimentation; multi-armed bandits; hypothesis testing; power analysis; reward--inference tradeoff}

\maketitle

\section{Introduction}
Hypothesis-testing–driven experiments are ubiquitous across scientific domains, including clinical trials \citep{austrian2021cds}, psychology \citep{cohen1988statistical}, education \citep{sales2023auxiliary}, and biology \citep{zar2010biostatistical}.
Scientists in these domains are often primarily concerned with statistically significant differences between interventions, for instance, to argue that an intervention is superior to a control. 
When designing experiments, scientists often use power analysis to determine how many samples to collect for a chosen hypothesis test (e.g., a $t$-test) and target type~I and type~II error rates \citep{cohen1988statistical,casella2002statistical,faul2009gpower}, and then collect that many samples using uniform randomization (UR). However, this uniform approach may allocate many samples to interventions which result in very poor outcomes (for instance, harming participants~\cite{rosenberger2016randomization} or resulting in a loss of revenue). The multi-armed bandit (MAB) framework offers a promising alternative by allocating participants to better-performing interventions (\textbf{arms}) and improving cumulative outcomes (\textbf{reward}) \citep{thompson1933likelihood,lattimore2020bandit}. 

However, MAB algorithms collect samples \textit{adaptively}, that is, the arm selected for sample $t$ depends on the reward collected from samples $1\dots t-1$, violating the assumptions of classic statistical tests and causing them to not be statistically valid in this setting~\citep{villar2015multiarmed,deshpande2018accurate,hadad2021confidence,smith2018bayesian}. Practically, this leads to inflated type~I and type~II errors, incentivizing scientists to largely avoid adaptive sampling despite the potential for greatly increased reward. Prior work has addressed this problem typically through generating specialized tests~\cite{deliu2021efficient} or bandit algorithms~\cite{xiang2022mabvsab}, which may be undesirable for practitioners who wish to use familiar and well-studied tests and algorithms. General-purpose solutions have been understudied - to the best of our knowledge, the Adaptive Randomization Test (ART) \citep{ham2023art} is the only correction approach that theoretically generalizes to all bandit algorithms and tests. However, there is no practical analysis of its performance with common hypothesis tests and multi-armed bandit algorithms, and as we show in Section~\ref{sec:result}, it can have exceptionally low statistical power for some popular algorithms and tests, leading it to be unattractive to practitioners.  There is a lack of practical and statistically sound alternatives to address this problem.

Even if statistical tests are corrected to return valid results, there exists a fundamental trade-off between reward (achieved through more exploitation) and hypothesis testing efficiency (which requires more balanced exploration among arms). Na\"ively applying  popular bandit algorithms such as Thompson Sampling (TS) \citep{thompson1933likelihood,chapelle2011empirical,russo2018tutorial}, will require large numbers of additional samples (longer \textbf{horizon}) to achieve the desired power compared to the traditional UR approach~\citep{williams2021challenges}. There remains a lack of tools to help practitioners select bandit algorithms that balance reward and horizon according to their problem-specific preferences.

To address these gaps, we make three key contributions. First, we propose a general-purpose \emph{algorithm-induced} testing correction that constructs the null distribution of the test statistic by simulating data collection under the same adaptive algorithm. We show this approach consistently achieves higher power than ART, performing an order of magnitude better on deterministic algorithms such as UCB~\cite{auer2002finite}. Second, we develop an objective function for adaptive experiment evaluation that allows experimenters to explicitly adjust the priority between reward maximization and the number of samples required to achieve a desired statistical power. Third, we present an optimization framework and toolkit\footnote{\url{https://anonymous.4open.science/r/Bandit-Simulation-4B02}} that efficiently recommends the best bandit algorithm parameters and experiment length for the user's desired cost and statistical constraints.

\section{Problem Setup}\label{sec:problem_setup}

We consider an adaptive experimentation design problem in which the experimenter seeks to both maximize reward during the experiment and conduct valid hypothesis testing at its conclusion. In this section, we introduce the experimental setting, notation, and basic concepts that will be used throughout the paper. 





\paragraph{Basic MAB setup}
We start by introducing the basic multi-armed bandit (MAB) data collection setup.
Let $K$ be the number of distinct interventions/\textbf{arms} in an experiment. We consider a stochastic setting where each arm $a$ is associated with a fixed but unknown reward distribution $\nu_a$, and denote the arm mean reward as $\mu_a = \mathbb{E}[r_t \mid a_t = a]$. Let $\vec{\nu} = (\nu_1, \dots, \nu_K)$ denote the reward distributions for all arms. Prior to the experiment, the experimenter specifies an MAB algorithm $\pi$ and the experiment length $T$. At each time step $t = 1, 2, \dots, T$, the bandit algorithm $\pi$ selects an arm $a_t \in \mathcal{A} = \{1, \dots, K\}$ and observes reward $r_t \sim \nu_{a_t}$. We denote the random ordered tuple of the experiment history up to time $t$ as $h_t = (a_1, r_1, a_2, r_2, \dots, a_t, r_t)$, and let $\mathcal{H}_t$ be the (possibly infinite) set of all possible values of $h_t$.

\paragraph{Hypothesis testing}
Based on the MAB experiment data $h_T$, the experimenter wishes to conduct a hypothesis test between two competing assumptions about the reward distributions of all arms $\vec{\nu}$, $H_0$ against $H_1$ (often called the null and the alternative hypothesis, respectively). To conduct the test, the experimenter specifies a critical region $C \subset \mathcal{H}_T$, and they reject $H_0$ if $h_T \in C$, sometimes loosely referred to as determining statistical significance. The set $C$ should be defined in a way that $h_T \in C$ means there is strong statistical evidence against $H_0$ and in favor of $H_1$, although the details differ based on the chosen hypothesis test. 

In this paper, we focus on classical hypothesis tests such as $t$-tests, ANOVA, and Tukey’s test. While both the problem formulation and our proposed solution can be extended more broadly, we restrict attention to tests whose critical regions are specified through a test statistic function $S$ and take the form of one- or two-sided rejection regions. For example, in a two-arm $t$-test with hypotheses $H_0:\mu_1=\mu_2$ and $H_1:\mu_1\neq\mu_2$, the test statistic $S(h_T)$ is the standard $t$-statistic (or its normal approximation for large samples). A two-sided 95\% test rejects $H_0$ when $|S(h_T)|>1.96$, yielding the critical region $C=\{h_T : |S(h_T)|>1.96\}$.

Experimenters want to control Type~I error $\alpha$ and Type~II error $\beta$ for their test of interest, using the data adaptively collected for algorithm $\pi$. The Type~I error of a test, also known as the False Positive Rate (FPR), is the chance of rejecting $H_0$ when it is true, defined as $ \alpha = \Pr\!\left(h_T \in C \mid H_0, \pi, T\right)$. 
In contrast, the Type~II error is the probability of failing to reject $H_0$ when $H_1$ is true, defined as 
$\beta = \Pr\!\left(h_T \notin C \mid  H_1, \pi, T\right).$ We refer to power as the complement of the Type~II error, $1-\beta$.

\paragraph{Power Analysis}
Statistical tests are commonly used to analyze experimental data; however, practitioners often wish to determine the sample size required to achieve a desired power $1-\beta$ (at a fixed FPR $\alpha$) before the experiment. This process is known as power analysis.

Note that Type~II error (and power) depends on the distribution of the experiment history $h_T$ under $H_1$. Such a distribution varies across underlying arm distributions $\vec{\nu} \in H_1$. Therefore, when conducting a power analysis in practice, power is evaluated with respect to a specified restriction on $\vec{\nu}$, reflecting the experimenter’s design priorities when distinguishing $H_1$ from $H_0$~\citep{faul2009gpower}.
Continuing our $t$-test example, one may wish to guarantee 80\% power at 5\% FPR when the gap between arm means satisfies $|\mu_1-\mu_2|=0.1$. The output of the power analysis is the number of samples required to meet these constraints.




\paragraph{Our problems}

To summarize, at the beginning of an experiment, the experimenter specifies $H_0$, $H_1$, and desired error levels $(\alpha,\beta)$. They then seek to determine the algorithm $\pi$, a valid $C$, and the required $T$. Note that controlling the FPR and evaluating power for MAB data is more challenging than in traditional experiments. In this work, we consider two related problems:

Problem~1 (test correction): For any algorithm $\pi$ that collects data adaptively, we seek test-correction methods that can construct a valid critical region $C_{\pi,\alpha}$ controlling FPR at level $\alpha$. Such correction may be applied either to the classical test statistic $S$ or directly to the critical region $C$. Subject to valid FPR control, we further aim to achieve higher power for a given sampling budget and algorithm. We discuss this problem in Section~\ref{sec:problem_1_hypothesis_testing}.

Problem~2 (algorithm comparison): The experimenter seeks to select an algorithm $\pi$ that achieves higher reward while requiring fewer samples $T$ to reach a desired power level. Note that these objectives are often in tension across different choices of the algorithm $\pi$. In Section~\ref{sec:problem_2_reward_inference}, we first examine the associated power analysis considerations for MAB experiments, and then formalize this trade-off through an objective function to guide algorithm selection.


The problems are interconnected: a proper estimate of power (Problem~1) is required to subsequently analyze its trade-off with reward (Problem~2).

\section{Related Work}
\label{sec:related_work}

In this work, we discuss approaches for conducting valid hypothesis tests for adaptive experiments and their tradeoff with experimental outcomes (reward). To maximize the likelihood of real-world impact, we focus on an algorithm-agnostic framework that can be applied across commonly used adaptive multi-armed bandit algorithms and hypothesis tests. To the best of our knowledge, there is limited discussion of generalizable test correction methods for arbitrary algorithms and tests, nor is there a formalized objective function concerning such a general tradeoff. Hence, in this section, we review the most closely related work.

\paragraph{Hypothesis Testing with Adaptively Collected Data}
Given the variety of commonly used hypothesis tests and MAB algorithms, relatively little analysis-focused work attempts to provide a general framework that is agnostic to both algorithm and test choice. One line of work has focused on specially constructed tests for an adaptive setting.
\citet{deliu2021efficient} proposed an allocation-probability-based test statistic, which is restricted to TS. Other work has proposed testing approaches based on martingales~\cite{xu2021unified} and probability clipping~\cite{yao2021powerconstrained}, which guarantee statistical validity for all bandit algorithms, but rely on a restricted set of hypotheses and custom statistical tests.

Another line of work has discussed custom bandit algorithms that  achieve good statistical performance in certain settings. \citet{xiang2022mabvsab} proposed a fixed allocation strategy that guarantees statistical power at least as high as uniform randomization in two-armed experiments; their analysis is limited to binary rewards and two-arm settings. Building on this line of work, \citet{kong2024sequential} proposed an elimination-based MAB algorithm that achieves Type~I and Type~II error control for specified pairwise tests, with a focus on best-arm identification.
However, these methods do not allow full freedom in algorithm choice, nor generalization to all common statistical tests.

\citet{williams2021challenges} introduced the concept of critical region correction to generalize across bandit algorithms and tests; however, they considered only a two-armed setting in which the ground-truth reward distributions are known a priori, limiting its applicability in practice. Inspired by their work, we formalize a procedure that removes these practical constraints. Moreover, we provide theoretical justification and empirical evaluation demonstrating that our method achieves strong power while maintaining valid FPR control.

To the best of our knowledge, the Adaptive Randomization Test (ART)~\cite{ham2023art} is the only approach that can correct arbitrary tests run under arbitrary bandit algorithms in arbitrary environments. ART controls Type~I error by using bootstrap resampling to simulate the null distribution of the test statistic while fixing the observed reward sequence. While theoretical promising in terms of Type~1 error, no analysis of its Type~II error is provided. We evaluate this method empirically  in Section~\ref{sec:problem_1_hypothesis_testing} and find it to have extremely low power in some settings, limiting its practical applicability. 

\paragraph{Balancing Reward and Inference Objectives}
Another key distinction of our work compared to the testing-correction methods discussed above is that they do not formalize the trade-off  between reward and hypothesis testing or statistical inference more generally. 

We now discuss work that examines how to formulate related trade-offs between reward and quantities related to statistical inference. Some work \citep{zhong2023achieving,simchi-levi2023mabdesign} adopts a Pareto-frontier perspective on related problems, proposing algorithms that achieve Pareto optimality between reward and best-arm identification (BAI) or average treatment effect (ATE) accuracy, respectively. However, this differs fundamentally from our problem, which focuses on balancing reward against the finite horizon induced by statistical hypothesis testing under explicit Type~I and Type~II error constraints. Our problem is therefore less amenable to a Pareto analysis, which in the bandit literature is typically asymptotic rather than finite-horizon in nature.

Similarly, prior work \citep{liu2014trading,erraqabi2017trading} has proposed objective functions that use a weighting parameter to trade off scaled estimation error with reward. However, this does not easily translate to statistical hypothesis testing, especially in an interpretable fashion.

Lastly, we also want to highlight work that introduces the cost concept into MAB problem formulations \citep{xia2015budgeted,kanarios24cabai,bui11committing}. Though they do not trade off reward with inference, their ideas inform the construction of our objective function.

\section{Problem 1: Adjusting Efficiently for Hypothesis Testing Confidence}
\label{sec:problem_1_hypothesis_testing}



Hypothesis testing with MAB data presents unique challenges \citep{smith2018bayesian, ham2023art}. For classical hypothesis tests involving multiple arms, it is often assumed that the sampling follows a fixed (non-adaptive) ratio among arms over time. However, when data are sampled using MAB algorithms, the sampling ratio depends on the observed reward history. To see why this matters, consider the null setting with two identical arms. If one arm happens to look worse early on due to random noise, an adaptive algorithm will sample it less frequently. This ``freezes'' the arm at an underestimated value while the other arm continues to be sampled, artificially inflating the apparent gap. This can cause various statistical inference issues, including biased estimation of arm means, distorted test statistic distributions, and inflated Type~I error/FPR.

\begin{table}[h]
\centering
\small
\setlength{\tabcolsep}{4pt}
\caption{
FPR under the null hypothesis $H_0\!:\,\mu_1=\mu_2=\mu$ for varying values of $\mu$.
Results compare the classical one-sided two-sample $t$-test with a fixed critical threshold of $1.64$ (from the $z$-table) to its AIT-corrected counterpart.
}
\label{tab:t_test_classical_vs_AIT_fpr}
\begin{tabular}{c ccccc}
\toprule
True location of $\mu$ in $H_0$
& 0.1 & 0.3 & 0.5 & 0.7 & 0.9 \\
\midrule
FPR (AIT corrected)
& 0.052 & 0.050 & 0.050 & 0.049 & 0.050 \\
FPR (using z-table)
& 0.071 & 0.086 & 0.099 & 0.108 & 0.132 \\
\bottomrule
\end{tabular}
\end{table}

Table~\ref{tab:t_test_classical_vs_AIT_fpr} illustrates the FPR inflation when using a classical critical region for a one-sided $t$-test ($1.64$ using a $z$-table). 
Moreover, we note that the degree of FPR inflation varies with the specification of the null setting. This implies that the null distribution of a classical test statistic can change according to the exact specification of $\vec{\nu}$ under the null hypothesis $H_0$. Such dependency/sensitivity poses a challenge for applying classical tests to MAB data.


In this section, we seek to develop a method to correct the FPR of arbitrary statistical tests and bandit algorithms. Similar to the assumption made by \citet{ham2023art}, \textbf{we relax the problem and focus on controlling FPR on a subset of the given classical $H_0$}.


\begin{assumption}[Equal arm distributions under the null]\label{assump:equal_dist}
Under the $H_0$, all arms share the same reward distribution:
\begin{equation}\label{eq:adp_H0}
\nu_1 \stackrel{d}{=} \nu_2 \stackrel{d}{=} \cdots \stackrel{d}{=} \nu_K .
\end{equation}
\end{assumption}

This assumption restricts some hypotheses; however, it matches the common understanding of a null hypothesis $H_0$, that is, that under a null hypotheses all conditions are identical.  Under this assumption, we can use all data in $h_T$ collectively and use MLE (or bootstrap) to estimate the null hypothesis reward distribution $\nu_{H_0}$. In what follows, we state our proposed test correction procedure, assuming the estimation on $\nu_{H_0}$ is given. Later in this section, we discuss and empirically show that it can achieve good power while controlling the FPR reasonably well.



\paragraph{Our proposed test correction method} 
We propose a plug-in correction called algorithm-induced test (AIT) correction. It is illustrated for a one-sided test in Procedure~\ref{proc:alg_induced_correction}. Note that $\nu_{H_0}$ is estimated from the full collected dataset $h_T$ under Assumption~\ref{assump:equal_dist}. At a high level, given an MAB setting $(K, T, \pi, \vec{\nu} = (\nu_{H_0}, \ldots, \nu_{H_0}))$, our approach simulates the distribution of a given test statistic $S$. Note that even if the algorithm $\pi$ is deterministic, the repeated simulation is important given the stochasticity in rewards drawn from $\nu_{H_0}$.  Using that simulated distribution, we follow the critical region form (e.g., a one-sided interval for a one-sided test) to manually calibrate the thresholds for FPR control. The output of the test is a critical region for each timestep; in many cases the next step would be to compare the actual test statistic $S(h_T)$ with $C_T$ to test for statistical significance. In this paper, we focus on tests where the original critical region is a one- or two-sided continuous interval, although our approach can be extended more broadly.

A key design choice is that our approach uses exactly the same test statistic $S$ calculation and the same critical region form as the original test. One may question the reasonableness of doing so, since Table~\ref{tab:t_test_classical_vs_AIT_fpr} shows that adaptive experimentation can change the critical region significantly. 
In what follows, we discuss a scenario where our correction results in the most powerful test.

\begin{procedure}[H]
\caption{Algorithm-Induced Test Correction (AIT)}\label{proc:alg_induced_correction}
\begin{algorithmic}[1]
\State \textbf{Input:} Classical hypothesis $H_0$, $H_1$ with test statistic $S$; number of arms $K$; horizon $T$; (estimated) $H_0$ reward distribution $\nu_{H_0}$; algorithm $\pi$; FPR constraint $\alpha_0$; simulation repetitions $M$
\Procedure{AlgoInducedCrit}{$K,\, T,\, \nu_{H_0},\, S,\, \pi,\, \alpha_0,\, M$}
    \For{$n = 1$ to $M$}
        \State Run \textbf{Bandit}$(K, \vec{\nu}=(\nu_{H_0},...,\nu_{H_0}), T, \pi)$ and obtain $h_{n,T}$
        \State Record test statistic $s_{n,t} = S(h_{n,t})$ for $t=1,2,...,T$
    \EndFor
    \State Compute $q_t$ as the $(1-\alpha_0)$-th quintile from $\{s_{n,t}\}_{n=1}^M$ for $t=1,2,...,T$
\EndProcedure
\State \textbf{Output:} Adjusted critical region sequence (right tail) $\{C_t\}_{t=1}^T$ with $C_t = (q_t,\infty)$ for $t=1,2,...,T$
\end{algorithmic}
\end{procedure}



\paragraph{Most powerful test for simple hypotheses.}
We extend the notion of simple hypotheses to the MAB setting, where both the null and alternative hypotheses fully specify the reward distributions across arms, denoted by $\vec{\nu}_0$ and $\vec{\nu}_1$, respectively. For example, in a two-armed Bernoulli setting one may consider $H_0:\mu_1=\mu_2=0.5$ versus $H_1:\mu_1=0.6,\mu_2=0.4$. Given data $h_T$ collected over $T$ rounds, the classical likelihood ratio test (LRT) statistic is $\frac{\prod_t p(r_t \mid a_t, \vec{\nu}_1)}{\prod_t p(r_t \mid a_t, \vec{\nu}_0)}$. Our main theorem shows that applying AIT to this statistic yields the most powerful test  under an adaptive data collection process $\pi$.

\begin{theorem}[AIT-optimality of the LRT on simple hypotheses]
\label{thm:ait_lrt_optimal}
Let $\pi$ be an arbitrary MAB algorithm. For testing simple hypotheses, $\vec{\nu}_0$ against $\vec{\nu}_1$, using data collected under $\pi$, the test with critical region constructed from classical LRT,
\[
\mathcal{C}_k = \{ h_T:\frac{\prod_t p(r_t \mid a_t, \vec{\nu}_1)}{\prod_t p(r_t \mid a_t, \vec{\nu}_0)} > k \}
\]
is the most powerful test at level
\[
\alpha = \mathbb{P}\{ h_T \in \mathcal{C}_k \mid \pi, \vec{\nu}_0 \}.
\]
\end{theorem}

\paragraph{Proof sketch.}
Fix an adaptive algorithm $\pi$. By the Neyman--Pearson lemma, the most powerful level-$\alpha$ test for testing a simple null $\vec{\nu}_0$ against a simple alternative $\vec{\nu}_1$ is based on the likelihood ratio $r_\pi(h_T)=\frac{p(h_T\mid\vec{\nu}_1,\pi)}{p(h_T\mid\vec{\nu}_0,\pi)}$. Although this expression appears to depend on $\pi$, we can make use of the conditional independence structure of the MAB process and show that, whenever $p(h_T\mid\vec{\nu}_0,\pi)>0$, it can be written as
$r_\pi(h_T)=\frac{\prod_t p(r_t \mid a_t, \vec{\nu}_1)}{\prod_t p(r_t \mid a_t, \vec{\nu}_0)}$.
That is, the likelihood ratio reduces to the classical LRT statistic computed on the realized actions and rewards, and is invariant to the choice of $\pi$. Consequently, the classical LRT remains the most powerful test under adaptive data collection. Nevertheless, recalibration is required to ensure valid Type~I error control, since the sampling distribution of $h_T$ generally depends on $\pi$. A detailed proof is provided in Appendix~\ref{app:proof_ait_lrt_optimal}.

While this does not show our AIT approach is most powerful in all settings, it does show our approach of using the same form of critical region and test statistic is the most powerful approach for simple hypotheses.

\paragraph{Evaluation on t-tests.} Practitioners are often interested in more complex hypotheses - here we empirically evaluate the power of AIT in this setting. We consider a two-sample $t$-test with a composite hypothesis $H_0:\mu_1=\mu_2$ against $H_1:\mu_1 \neq \mu_2$, with total samples $T=200$. The two arms have Bernoulli reward with mean being $0.6$ and $0.4$, respectively. As we show in Table~\ref{tab:alginduced_vs_ART_TS}, our critical region construction yields better power than ART when applied to popular adaptive algorithms such as TS, $\epsilon$-greedy, and UCB. Moreover, we note that ART exhibits degenerate power for fully deterministic algorithms such as UCB~\citep{auer2002finite}. This is because it fixes the reward history while re-simulating arm selection; however, for a deterministic algorithm this is ineffective as it always results in identical histories for deterministic algorithms. 
In Appendix~Table~\ref{tab:more_art_ait_comp}, we evaluated it on more algorithms, and our AIT achieves consistently better power. Besides, even though we need to estimate the null, which can introduce estimation error, we find that AIT can control FPR reasonably well. As shown in Table~\ref{tab:t_test_classical_vs_AIT_fpr}, when conducting the test at the $95\%$ confidence level, the FPR is empirically controlled at $0.05 \pm 0.002$ for $H_0:\mu_1=\mu_2$ across null reward means ranging from $0.1$ to $0.9$.


\begin{table}[htbp]
\centering
\caption{
Power comparison between ART and AIT for data collected using different algorithms.
All tests are calibrated to have empirical FPR around $0.05$.
}
\label{tab:alginduced_vs_ART_TS}
\begin{tabular}{l cc cc}
\toprule
 & \multicolumn{2}{c}{\textbf{Power}} & \multicolumn{2}{c}{\textbf{FPR}} \\
\cmidrule(lr){2-3}\cmidrule(lr){4-5}
\textbf{Algorithm} & \textbf{ART} & \textbf{AIT} & \textbf{ART} & \textbf{AIT} \\
\midrule
TS              & 0.434  & 0.520  & 0.052 & 0.053 \\
$\epsilon$-greedy (0.1) & 0.443 & 0.490 & 0.057 & 0.057 \\
UCB             & 0.050  & 0.781  & 0.050 & 0.054 \\
\bottomrule
\end{tabular}
\end{table}

\section{Problem 2: Trading Off Reward and Statistical Power}
\label{sec:problem_2_reward_inference}

The approach discussed in Section~\ref{sec:problem_1_hypothesis_testing} provides experimenters with a way to conduct valid hypothesis testing for data collected under any given policy $\pi$.
We now turn to power analysis in the MAB setting, and to the problem of balancing reward maximization against the sample size required to achieve a desired level of statistical power.

To begin, we give a concrete definition of power analysis for MAB algorithms.
In contrast to classical power analysis, which typically evaluates power under a fixed alternative (e.g., a specified effect size), we adopt a distributional specification over reward parameters.
Specifically, we allow the experimenter to specify a prior distribution $D_{\vec{\nu}}$ over $\vec{\nu}$.
This choice serves two purposes.
First, such specifications are common in bandit modeling and are required to simulate rewards and evaluate average reward performance.
Second, the distribution of test statistics—and hence power—can be highly sensitive to the underlying reward parameters, making an explicit specification necessary for power evaluation in adaptive experiments.
We acknowledge this reliance on distributional assumptions as a limitation and an open challenge in MAB experimentation.

Power is then defined as
\[
\mathbb{P}(h_T \in C \mid D_{\vec{\nu}}, \pi, T).
\]

However, in many scientific, medical, and business applications, practitioners not only desire to adjust experiment length to the minimum necessary to achieve statistical power, but they also desire to allocate samples to inventions in a way that maximizes benefit and  minimizes undesirable outcomes.
A key difficulty is that these two goals are in some sense opposed to teach other: adaptive algorithms which are better at maximizing reward take much longer to reach the same level of statistical confidence. 
For example, in a two-armed binary-reward case with an arm difference of $0.2$, UR requires roughly $200$ samples to achieve $0.8$ power (fixing the Type~I error at $0.05$), whereas TS can require around $2000$ samples, despite getting much better average reward than UR (even in the first 200 samples). 
This highlight the need for a unified objective function to help practitioners navigate the trade-off between reward and horizon, and determine which algorithm best fits their setting.



\paragraph{Developing an objective function} For a given algorithm with $T$ samples required in the power analysis and cumulative reward $R$, we seek to construct an objective function $F(T, R, w)$ that takes as input $T$, $R$, and a single interpretable parameter $w$.  We refer to $w$ as the \textbf{experiment extension cost}, which captures the cost of involving an extra step or participant in the experiment in units of cumulative reward.\footnote{This is similar to ``cost of experimentation'' proposed in past work~\cite{bui11committing}, albeit in a different setting without statistical constraints.}  Higher values of $w$ mean extending the experiment is costly, and the user would prefer a shorter experiment, whereas lower values of $w$ mean extending an experiment is inexpensive, and improving reward is more important.    The chief challenge of specifying $w$ is in the conversion between reward units and cost to extend the experiment - in some business domains these may be in the same unit, but in scientific settings, this is often not the case. Thus, we provide tools and visualizations to help with this problem (Section~\ref{sec:choosing_w}).  


Given a user-specified $w$, we want this to represent the minimum \textbf{improvement} in cumulative reward needed to trade off against the cost of an additional experimental step.\footnote{For any particular problem, at some point it may not be possible for a bandit algorithm to achieve a $w$ increase in cumulative reward—at which point the horizon will dominate the objective. This is desirable, as it matches the intuition that we want a short experiment with maximal cumulative reward.} However, if an experiment is extended by one step without improving the mean reward ($= R/T$), then its cumulative reward would normally increase from $R$ to $R \cdot \frac{T+1}{T}$. We argue that this amount should not be viewed as an \emph{improvement} in reward, but rather as a no-improvement \emph{increase} in horizon, and should therefore not improve the value $F$. Otherwise, extending the horizon to achieve the same average reward would favor long, rather than short, experiments. Given this, we derive the partial derivative equations of the characteristic function that removes such effect. Assume we have two algorithms $\pi$ and $\pi'$, with horizons $T$ and $T+1$ and cumulative reward $R$ and $R'$, respectively. Then, restating the desired property of $F$ formally states that: \\ when $R' - \frac{T+1}{T}R = w$, then $F(T+1, R', w) \approx F(T, R, w)$. 
We can rewrite the former equation as  $R' - R - \frac{(T+1)-T}{T}R = ((T+1)-T)w$. 
If we let $\Delta R = R' - R$ and $\Delta T = (T+1)-T$, we can rewrite this as follows: We require 
 $F(T, R, w) \approx F(T+\Delta T, R + \Delta R, w)$ when
\begin{align}
\Delta R - \frac{R}{T}\,\Delta T = w\,\Delta T 
\label{eq:discrete_relationship}
\end{align}

Now we take a continuous and unbounded relaxation of $T$ and $R$, and convert the above analysis into a partial differential equation,  By the property of the total differential of a multivariable function, we rewrite $F(T+\Delta T, R + \Delta R, w)$ in a continuous form as: \\$dF\big|_{w = w_0}
= \frac{\partial F}{\partial T}\, dT
+ \frac{\partial F}{\partial R}\, dR .$
And equation~\eqref{eq:discrete_relationship} now becomes: \\$dR - \frac{R}{T}\, dT = w_0\, dT$. That is, we expect $dF = 0$ when \\$dR = (R/T + w_0)\, dT$. Substituting $dR$ into $dF$, we have:
$$
dF\big|_{w = w_0}
= \frac{\partial F}{\partial T}\, dT
+ \frac{\partial F}{\partial R}\, (R/T + w_0)\, dT = 0 .
$$
Factoring out the $dT$ term and observing that the remaining factor must be zero gives our main condition (equation~\eqref{eq:main_princ} below).

1. \textbf{(Main principle) Iso-value trade-off condition} For any fixed $w_0$, we expect $F(T, R, w_0)$ to encode the empirical tradeoff between $T$ and $R$ such that the following partial differential equation (PDE) is satisfied:
\begin{equation}
\frac{\partial F}{\partial T}
+ \frac{\partial F}{\partial R}\, (R/T + w_0) = 0 .
\label{eq:main_princ}
\end{equation}

\paragraph{Our proposed objective function} 
To this end, we propose the use of the following objective function to measure and compare the reward--inference tradeoff for different choices of algorithm $\pi$.  
\begin{equation}\label{eq:main_objective}
F(T,R,w) = R/T - w \cdot \log(T).
\end{equation}


It can be easily verified that our proposed function satisfies the partial equation mentioned earlier. Indeed, $\partial F/\partial T = -R/T^2 - w/T$ and $\partial F/\partial R = 1/T$, so
$\partial F/\partial T + (\partial F/\partial R)(R/T + w) = 0$.
We call our proposed objective function \textbf{experiment-cost-penalized reward} (ECP-reward), as it contains the commonly used average reward metric in the bandit literature and penalizes it with $\log$-scaled experiment step costs. Before giving an in-depth analysis of the proposed function, we try to give a quick interpretation of it.

The behavior of $F$ is fairly intuitive: When $w = 0$, the objective reduces to maximizing experiment mean reward, and when $w$ is very large,  the objective reduces to minimizing experiment steps.



Other than the main principle, our proposed formula also satisfies the following useful properties: 

2. \textbf{Monotonicity.} Fix $R$ and assume $w>0$. Then the objective score decreases monotonically as horizon increases: $\partial F / \partial T < 0$. Similarly, fixing $T$ and $w$, we have $\partial F / \partial R > 0$.  

3. \textbf{Consistency under reward location and scale shifting.} The relative ordering of two experiment scores remains unchanged if we shift the reward unit/basis by a constant or scale it by a positive factor. Formally, if $F(T_1,R_1,w) > F(T_2,R_2,w)$, then for any intercept $b \in \mathbb{R}$ and any multiplier $a>0$, we have $F(T_1,\, R_1 + b\,T_1,\, w) \;>\; F(T_2,\, R_2 + b\,T_2,\, w),$ and $F(T_1,\, a R_1,\, a w) \;>\; F(T_2,\, a R_2,\, a w).$\\



Lastly, one might wonder why we do not consider simpler objectives such as $R - wT$ or $R/T - wT$. Essentially, these objectives do not follow our main principle and can lead to unituitive evaluations. For further discussion, see Appendix~\ref{app:why_naive}.

Now that we have defined an objective, we can optimize for it, finding the best bandit algorithm (and parameters, if any) for the user's chosen $w$. An optimization procedure based on this objective is provided in Procedure~\ref{proc:optimization} in the Appendix.

\section{Simulation Method}\label{sec:sim_method}

In our simulation study,  we show how our testing correction method and objective function can optimize the experimental design while retaining statistical validity. We focus on simple, widely used algorithms and tests to ensure interpretability and practical relevance. 

\paragraph{Computational cost reduction}
Unlike a bandit simulation where computation is not a primary concern, power analysis for bandit algorithms can be computationally costly. It involves nested simulations, with an additional Monte Carlo loop to calibrate the test for each simulated experimental run (i.e., Procedure~\ref{proc:alg_induced_correction} is required within each simulation replication). We note that a similar issue exists for the ART correction proposed by \citet{ham2023art}.

To improve simulation efficiency, we employ several techniques, including vectorization, batched algorithm updates with a geometrically increasing batch size, result aggregation within each batch, and binning similar estimated null hypotheses for shared calibration. Empirically, this greatly reduces time - while a naive implementation of a single configuration with $N=1000$ (assuming additional $M=500$ without approximation) and horizon $T=200$ can take an hour to run and require roughly $20$~GB of memory, whereas our approach completes it in under one second with negligible memory usage. See Appendix~\ref{appendix:approximation} for more details.

\paragraph{Algorithms}
We select Thompson Sampling (TS) and Uniform Randomization as baselines. TS follows the classical Bayesian formulation: for Bernoulli rewards we use a Beta--Bernoulli model with a Beta$(1,1)$ prior, and for Gaussian rewards we adopt the conjugate Normal--Inverse-Gamma update \citep{thompson1933likelihood,gelman2013bayesian} with a non-informative prior. We also include $\epsilon$-Thompson Sampling ($\epsilon$-TS), which augments TS with fixed uniform exploration by following TS with probability $1-\epsilon$ and selecting an arm uniformly at random with probability $\epsilon$, where $\epsilon\in[0,1]$ controls the exploration rate.

\paragraph{Empirical study inspired simulation}
To illustrate our framework in a scientific setting, we first consider a six-armed bandit environment derived from a large-scale online educational experiment~\citep{Reza2023ExamEustress} in which psychologists investigated how brief, digitally delivered stress-reappraisal activities influence students’ learning and exam performance. 
The experiment showed different types of messages which varied the amount of explanatory text, modality (text/video), and the presence of reflection prompts.  

When simulating how bandit algorithms perform in this setting, we use two phases:  we first perform a pre-experiment optimization based on the objective function to determine the optimal algorithm parameter and horizon, similar to power analysis.
We then use the realized arm means from the empirical data in a post-experiment simulation to validate the practical performance of the recommended parameter choice.

In the pre-experiment phase, we scale exam scores to $[0,1]$ and assume the prior mean reward for each arm follows
$\mathrm{N}(0.81, 0.015^2)$. 
In each simulation for parameter optimization, the arm-level mean $\mu_i$ is drawn from this prior, and rewards follow
$\mathrm{N}(\mu_i, 0.1^2)$. Those values are estimated and rounded from the arm mean rewards, their average standard deviations and their prediction residuals, respectively from the original data used by \citet{Reza2023ExamEustress}.
For the hypothesis-testing framework, we conduct independent pairwise comparisons between each arm and a fixed control using two-sided $t$-tests, treating the first arm (corresponding to the pure text summary intervention) as the control.
We target a Type~I error rate of $0.05$, a Type~II error rate of $0.2$, and a minimum detectable effect size of $0.025$. 
We assume the experimenter specifies $w = 0.01$, which can be interpreted as valuing one additional experimental step
as costing roughly 1 exam point out of 100. 
This choice reflects the empirical finding from the original study that large treatment effects are unlikely, while even small improvements remain practically meaningful.


In the second phase,  we examine a possible scenario if the chosen algorithm was deployed with the chosen horizon. Instead of defining a prior distribution and sampling arm means, for this phase we use the actual arm means suggested by the data, $\vec{\mu}=$[0.810, 0.806, 0.819, 0.778, 0.827, 0.813], and maintain reward distribution as $\mathrm{N}(\mu_i, 0.1^2)$. We fix the experimental horizon using the pre-experiment analysis and assess performance by recomputing the achieved average reward and ECP-reward.
We highlight that this second simulation is illustrative, reflecting what may occur when the realized environment lies within the range specified by the prior.
We compare the resulting performance against baseline designs using TS and UR.

\paragraph{Additional simulation across tests}
We next apply our optimization procedure to a broader set of hypothesis tests in a three-armed binary-reward setting. 
Specifically, we consider ANOVA to assess whether mean rewards differ across arms in aggregate, pairwise t-tests (t-constant and t-control) to compare individual arms either against a fixed baseline or a designated control arm, and multiple-comparison procedures based on Tukey’s test to determine whether a single arm is distinctly superior. Detailed definitions of these tests are provided in Appendix~\ref{app:tests}.

In this setting, each arm’s true mean reward $p_i$ is drawn independently from a $\mathrm{Beta}(5,5)$ prior, and rewards follow $r_i \sim \mathrm{Bernoulli}(p_i)$. 
For each combination of bandit algorithm and parameter choice, results are averaged over $10{,}000$ replications.

\section{Result}
\label{sec:result}

\subsection{Empirical example inspired simulation}

\begin{figure}[ht]
    \centering
    \caption{
    Screenshot of our optimization framework web application, showing the relative ECP-reward performance
    for the empirical study inspired simulation. Note the best setting for $\epsilon$-TS outperforms TS and UR near the $w=0.01$.
    }\label{fig:EpsTS_ANOVA_objective_score}
    \includegraphics[width=\linewidth]{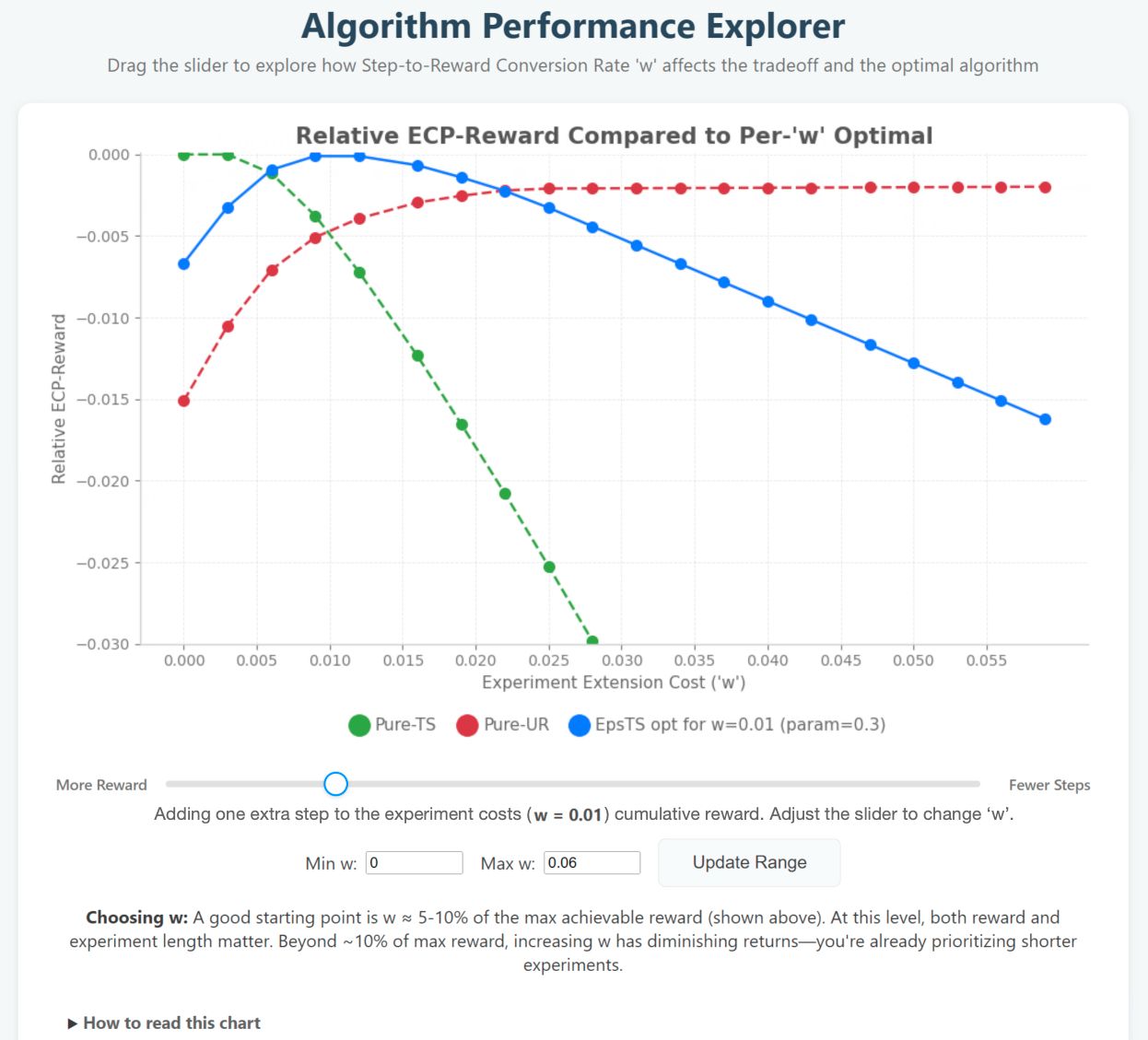}
\end{figure}

In this section, we present simulation results based on the empirical setting described in Section~\ref{sec:sim_method} and illustrate how experimenters can use our framework to guide their design of adaptive experimentation. Figure~\ref{fig:EpsTS_ANOVA_objective_score} illustrated the analysis visualizations displayed in our GUI output page. It allows the experimenter to identify the best algorithm and parameter (among what they considered), and compare its ECP-reward performance against UR and TS benchmark across various values of $w$. See Appendix~\ref{app:user_input} for an illustration of the user input page of our GUI. We now explain the improvement an experiment design can achieve after using our procedures.

\paragraph{Test correction.}
We first demonstrate that power analysis for MAB algorithms must account for test correction.
In our empirically inspired simulation, using an uncorrected two-sided $t$-test for each arm--control comparisons suggests that Thompson Sampling requires roughly $2{,}800$ participants to achieve power $0.8$, but results in an inflated false positive rate of $0.072$.
As shown in Table~\ref{tab:merged_prior_post}, correcting for algorithm-induced Type~I error using our AIT procedure increases the required sample size to over $4{,}000$ participants for the same target power, potentially altering experimenters’ design choices.

\begin{table}[t]
\centering
\small
\setlength{\tabcolsep}{2.5pt}
\renewcommand{\arraystretch}{1.05}
\caption{Comparison of experiment designs in an empirically inspired simulation with the \textbf{experiment extension cost} $w=0.01$.  Prior columns show expected performance under the prior specification at design time, while post columns report realized performance using empirical arm means. Underlined entries indicate values that can be corrected or optimized. Standard errors for all reward and ECP-reward values are below $0.0003$.}
\label{tab:merged_prior_post}
\begin{tabular}{l c c cc cc}
\toprule
 &  &  & \multicolumn{2}{c}{\textbf{Prior (design-time)}} & \multicolumn{2}{c}{\textbf{Post (realized)}} \\
\cmidrule(lr){4-5}\cmidrule(lr){6-7}
\textbf{Algorithm/Design} & \textbf{FPR} & \textbf{Steps} & \textbf{Reward} & \textbf{ECP} & \textbf{Reward} & \textbf{ECP} \\
\midrule
UR (naive) & 0.050 & 906 & 0.8100 & \underline{0.7419} & 0.8053 & \underline{0.7372} \\
\addlinespace[1pt]
TS (naive) & \underline{0.072} & 2{,}767 & / & / & / & / \\
\addlinespace[1pt]
TS (AIT) & 0.050 & 4{,}186 & 0.8251 & \underline{0.7417} & 0.8255 & \underline{0.7421} \\
\addlinespace[1pt]
$\varepsilon$-TS (AIT + Opt.) & 0.050 & 1{,}338 & 0.8185 & 0.7465 & 0.8162 & 0.7443 \\
\bottomrule
\end{tabular}
\end{table}

\paragraph{Optimization on experiment design.}
Next, we compare naive experimental designs (e.g., using the traditional TS or UR algorithms) with designs produced using our optimization procedure.  Figure~\ref{fig:EpsTS_ANOVA_objective_score} visualizes the relative ECP-reward of algorithms within the $\varepsilon$-TS family.
For each value of the \textbf{experiment extension cost} $w$, we define the relative ECP-reward of an algorithm $\pi$ as
$F(T_\pi, r_\pi, w)- F(T_{\pi^*}, r_{\pi^*}, w)$,
where $\pi^*$ denotes the best-performing algorithm among the candidates under consideration. In our setting, this corresponds to the optimal choice of $\epsilon$ within the $\varepsilon$-TS family, and a relative ECP-reward of $0$ indicates the optimal design.


The interface allows experimenters to specify $w$ via a slider, highlights the corresponding optimal $\varepsilon$-TS variant, and displays its relative ECP-reward across other values of $w$, together with the UR and TS benchmarks for comparison.

\paragraph{How experimenters can use our GUI to inform design}
\label{sec:choosing_w}
The TS and UR benchmarks help experimenters make core design decisions by clarifying how the relative importance of experimental cost affects algorithm choice.
Figure~\ref{fig:EpsTS_ANOVA_objective_score} shows that TS (corresponding to $\varepsilon=0$) is optimal only when $w$ is very small, indicating a regime in which experimental cost is effectively negligible and the objective reduces to near–pure reward maximization.
As $w$ increases, the relative ECP-reward of TS drops rapidly, signaling diminishing suitability when experimental cost becomes non-negligible.

For larger values of $w$, minimizing the total number of experimental steps becomes the dominant consideration.
In this region, UR performs comparatively well, although it is not always strictly optimal.
In our setting, for example, $\varepsilon$-TS with $\varepsilon=0.8$ achieves a slightly lower total sample size, while the ECP-reward of UR remains high, indicating that modest deviations from uniform sampling toward better-performing arms can sometimes yield hypothesis testing efficiency gains.

In this intermediate regime, hybrid $\varepsilon$-TS algorithms that balance reward maximization and exploration achieve better trade-offs.
The visualization allows experimenters to identify the optimal $\varepsilon$ for a given $w$ and to assess the robustness of that choice across nearby values.
For example, at $w=0.01$, the optimal variant is $\varepsilon$-TS$(0.3)$, which outperforms both TS and UR over a neighborhood of $w=0.01$.
As shown in Table~\ref{tab:merged_prior_post}, $\varepsilon$-TS$(0.3)$ achieves an ECP-reward of $0.7465$, approximately $0.4$ higher than either TS or UR, at the cost of roughly $400$ additional participants compared to UR while yielding an average reward increase of about $0.8$ points.

Overall, the visualization guides algorithm selection by showing when each design is preferable and how much value it provides.

\paragraph{Post-Experiment Evaluation.}
We next evaluate the optimized design using the realized experimental data.
Based on the pre-experiment analysis, we select $\varepsilon = 0.3$ for the $\varepsilon$-TS family when $w = 0.01$, with the total number of experimental steps fixed from the pre-analysis stage.
Using bootstrap resampling of the empirical data, we update the achieved average reward and recompute the corresponding ECP-reward, comparing performance against TS and UR.

As shown in Table~\ref{tab:merged_prior_post}, the optimized $\varepsilon$-TS$(0.3)$ design maintains a higher ECP-reward than both TS and UR.  Our optimized choice took 2,848 fewer steps to reach the desired Type~II error than TS, while achieving 0.0109 more average reward than UR, choosing an algorithm that better optimized the desired tradeoff. 
We note that, as is common in power analysis, the realized performance depends on how closely the true environment matches the prior assumptions, and the recommended parameter is not guaranteed to be optimal in every instance.

\subsection{Additional Evaluation Across Tests}\label{sec:result_objective_function}

In this section, we present simulation studies across a range of hypothesis tests and bandit algorithms.
We show that Procedure~\ref{proc:optimization} can improve upon naive experiment designs and that it performs robustly across a wide variety of tests.
Throughout, we fix the Type~I error rate at $0.05$ and the power constraint at $0.8$ for all tests, and assume a experiment extension cost $w$ of $0.1$.
According to our definition to $w$ in Section~\ref{sec:problem_2_reward_inference}, this choice reflects the interpretation that one additional step of the experiment  costs 0.1 in additional cumulative reward, which is a realistic choice.

As illustrated in Table~\ref{tab:objective_scores_framework}, compared to traditional fixed designs such as TS, UR, or $\varepsilon$-TS with $\varepsilon=0.5$, applying our optimization procedure can improve the achievable ECP-reward.
In particular, we observe relatively large improvements over TS and UR, exceeding $0.06$ on average. 

More importantly, each fixed choice of algorithm can exhibit substantial suboptimality in certain scenarios.
For example, UR and $\varepsilon$-TS$(0.5)$ perform poorly under the Tukey or T-Constant test, while TS shows particularly low ECP-reward for the (two-sided) T-control test.
These results highlight the importance of explicitly specifying the reward--inference tradeoff and systematically performing algorithm and parameter selection.

\begin{table}[htbp]
\centering
\small
\setlength{\tabcolsep}{7pt}
\caption{
ECP-reward comparison (higher is better) across fixed experiment designs (TS, UR, and $\varepsilon$-TS with $\varepsilon=0.5$) versus our optimized design over the $\varepsilon$-TS family using Procedure~\ref{proc:optimization}.
}
\label{tab:objective_scores_framework}

\begin{tabular}{lcccc}
\toprule
 & ANOVA 
 & \makecell{T-Constant} 
 & \makecell{T-Control} 
 & \makecell{Tukey} \\
\midrule
Naive UR              
& -0.052 & -0.012 & -0.077 & -0.112 \\
Naive TS              
& -0.079 & 0.075 & -0.163 & -0.016 \\
Naive $\varepsilon$-TS(0.5)     
& -0.012 & 0.042 & -0.046 & -0.021 \\
\midrule
Optimized $\varepsilon$-TS   
& -0.009 & 0.075 & -0.042 & 0.012 \\
(Best $\epsilon$)     
& (0.4) & (0) & (0.4) & (0.2) \\
\bottomrule
\end{tabular}
\end{table}

\subsection{Prior mis-specification sensitivity check}\label{sec:result_objective_function_sensitive}

\begin{table}[t]
\centering
\small
\setlength{\tabcolsep}{3.5pt}
\caption{
ECP-reward loss relative to the true optimal per setting.
\textbf{Mis-Opt.}: loss from optimizing under a mis-specified prior;
\textbf{Rand. Baseline}: loss averaged over random picks on $\epsilon$ for $\epsilon$-TS.
}
\label{tab:prior_misspec}
\begin{tabular}{lccccccc}
\toprule
\multicolumn{8}{c}{\textbf{Prior location mis-specification (fixed scale = 0.15)}} \\
\addlinespace[1pt]
True location & 0.20 & 0.25 & 0.30 & \textbf{0.35} & 0.40 & 0.45 & 0.50 \\
\specialrule{0.25pt}{1pt}{1pt}
Mis-Opt. & 0.047 & 0.024 & 0.001 & 0.000 & 0.002 & 0.000 & 0.001 \\
Rand. Baseline & 0.083 & 0.050 & 0.018 & 0.017 & 0.019 & 0.018 & 0.019 \\
\specialrule{0.35pt}{2pt}{2pt}
\specialrule{0.35pt}{2pt}{2pt}
\multicolumn{8}{c}{\textbf{Prior scale mis-specification (fixed location = 0.35)}} \\
\addlinespace[1pt]
True scale & 0.09 & 0.11 & 0.13 & \textbf{0.15} & 0.17 & 0.19 & 0.21 \\
\specialrule{0.25pt}{1pt}{1pt}
Mis-Opt. & 0.010 & 0.004 & 0.000 & 0.000 & 0.003 & 0.005 & 0.005 \\
Rand. Baseline & 0.019 & 0.017 & 0.016 & 0.017 & 0.021 & 0.025 & 0.037 \\
\bottomrule
\end{tabular}
\end{table}

When the location and/or scale of the prior distribution over arm mean rewards is inaccurately specified, the optimization outcome may degrade. To assess the robustness of our experiment design framework to prior mis-specification, we evaluate how the objective score changes under alternative prior specifications.

We study both location and scale mis-specification for the $\epsilon$-TS design under an ANOVA test with Type~I error bounded by $0.05$ and power constraint $0.8$. Throughout, the algorithm is optimized assuming a $Beta(3.2,5.9)$ prior on each arm mean reward $\mu_i$ (mean $0.35$, standard deviation $0.15$), with the experiment extension cost $w=0.1$. For location mis-specification, the true prior mean varies while the scale is fixed; for scale mis-specification, the true prior standard deviation varies from $0.09$ to $0.21$ with the mean fixed.

Table~\ref{tab:prior_misspec} reports ECP-reward loss relative to the true optimal at each setting. We include a random-$\epsilon$ baseline to provide a reference scale. Relative to this baseline, the loss induced by prior mis-specification is small across a wide range of mean and scale deviations. The loss is largest when the true prior deviates most from the reference specification, such as at true location $0.2$ or true scale $0.09$, where the relative mismatch (ratio) is the most extreme. Overall, it indicates that the proposed optimization framework yields $\epsilon$ values that are robust to moderate prior uncertainty.

\section{Discussion and Conclusion}\label{sec:discussion}


In this work, we provide practical and justifiable approaches to help practitioners design statistically reliable and cost-effective adaptive experiments, balancing the trade-off between cumulative reward and experimental deployment cost. Our approach consists of two main components. First, we introduce a valid hypothesis-testing correction method for adaptively collected data that can be easily generalized to commonly used tests such as ANOVA, Tukey, and the $t$-test. Compared to existing general-purpose adaptations, our approach has theoretically justifiable power efficiency and is empirically superior. 
Second, we propose a principled objective-function design that allows practitioners to evaluate two seemingly incompatible metrics—reward and hypothesis-testing efficiency—on a unified and at least somewhat interpretable scale. 
As we showed in Tables \ref{tab:merged_prior_post}, this can result in choosing an algorithm that achieves much better reward than UR while reducing the number of steps compared to TS significantly, therefore better optimizing the user's intended tradeoff.
Finally, we implement these methods into a unified software toolkit, offering practitioners a practical, one-stop solution for setting up and evaluating adaptive experimental designs, with the access to their familiar tests and algorithms.

\subsection{Limitation and future directions}
We discuss limitations and future directions below.

\paragraph{Usability}
While we presented interfaces and tools that allow users to explore parameter choices, we acknowledge that these require training to use effectively. One important direction of future work is to study which methods of training and interface design make these interfaces the most accessible to practitioners, which would likely involve user studies. Nevertheless, designing the objective scores, framework, and toolkit is an important first step to making statistically-sound adaptive experimentation more accessible to a broader community of scientists and experimenters.

\paragraph{Theoretical analysis}
While our framework produces recommended parameter choices that are approximately optimal, one limitation is that we did not show formal regret bounds (with respect to the ECP score) on the performance of these algorithms. One challenge in doing this is that a traditional worst-case analysis is not wholly appropriate, as performance is expected to be related to the prior distribution specified by the user. As such, as Bayesian regret bound~\cite{osband2013more} may be more appropriate.  One benefit of such analysis is that it would require a tight characterization of the estimation error introduced by our approximations.

\paragraph{Bayesian hypothesis testing}

In this paper, we have assumed classical statistical tests and classical power analysis. It would be valuable to extend our extend our analysis to the realm of Bayesian statistics.


\bibliographystyle{plainnat}   
\bibliography{reference/reference}

@book{casella2002statistical, title={Statistical Inference}, author={Casella, George and Berger, Roger L}, edition={2}, publisher={Duxbury Press}, year={2002}}

@book{cohen1988statistical,
  title        = {Statistical Power Analysis for the Behavioral Sciences},
  author       = {Cohen, Jacob},
  edition      = {2},
  year         = {1988},
  publisher    = {Lawrence Erlbaum Associates},
  address      = {Hillsdale, NJ}
}

@article{faul2009gpower,
  title   = {Statistical power analyses using {G*Power} 3.1: Tests for correlation and regression analyses},
  author  = {Faul, Franz and Erdfelder, Edgar and Buchner, Axel and Lang, Albert-Georg},
  journal = {Behavior Research Methods},
  volume  = {41},
  number  = {4},
  pages   = {1149--1160},
  year    = {2009}
}

@Book{lattimore2020bandit,
  author    = {Lattimore, Tor and Szepesvári, Csaba},
  title     = {Bandit Algorithms},
  publisher = {Cambridge University Press},
  year      = {2020}
}

@article{austrian2021cds,
  author = {Austrian, Jonathan and Mendoza, Felicia and Szerencsy, Adam and Fenelon, Lucille and Horwitz, Leora I. and Jones, Simon and Kuznetsova, Masha and Mann, Devin M.},
  title = {Applying A/B Testing to Clinical Decision Support: Rapid Randomized Controlled Trials},
  journal = {Journal of Medical Internet Research},
  year = {2021},
  volume = {23},
  number = {4},
  pages = {e16651},
  doi = {10.2196/16651}
}

@article{sales2023auxiliary,
  author = {Sales, Adam C. and Prihar, Ethan B. and Gagnon-Bartsch, Johann A. and Heffernan, Neil T.},
  title = {Using Auxiliary Data to Boost Precision in the Analysis of A/B Tests on an Online Educational Platform: New Data and New Results},
  journal = {Journal of Educational Data Mining},
  year = {2023},
  volume = {15},
  number = {2},
  pages = {53--85},
  doi = {10.5281/zenodo.8016854}
}

@InProceedings{liu2014trading,
  author    = {Liu, Yun-En and Mandel, Travis and Brunskill, Emma and Popović, Zoran},
  title     = {Trading Off Scientific Knowledge and User Learning with Multi-Armed Bandits},
  booktitle = {Proceedings of the Educational Data Mining Conference (EDM) 2014},
year      = {2014}
}

@Article{thompson1933likelihood,
  author    = {Thompson, William R.},
  title     = {On the Likelihood that One Unknown Probability Exceeds Another in View of the Evidence of Two Samples},
  journal   = {Biometrika},
  volume    = {25},
  number    = {3/4},
  pages     = {285--294},
  year      = {1933}
}

@article{russo2018tutorial, title={A tutorial on Thompson sampling}, author={Russo, Daniel J and Van Roy, Benjamin and Kazerouni, Abbas and Osband, Ian and Wen, Zheng}, journal={Foundations and Trends in Machine Learning}, volume={11}, number={1}, pages={1--96}, year={2018}}

@InProceedings{chapelle2011empirical,
  author    = {Chapelle, Olivier and Li, Lihong},
  title     = {An Empirical Evaluation of Thompson Sampling},
  booktitle = {Advances in Neural Information Processing Systems (NeurIPS) 2011},
year      = {2011},
  url       = {https://papers.nips.cc/paper/2011/file/9e082b752ed5630062c1b4b877e5f38b-Paper.pdf}
}

@InProceedings{erraqabi2017trading,
  author    = {Erraqabi, Akram and Lazaric, Alessandro and Valko, Michal and Brunskill, Emma and Liu, Yun-En},
  title     = {Trading off Rewards and Errors in Multi-Armed Bandits},
  booktitle = {Artificial Intelligence and Statistics (AISTATS) 2017},
year      = {2017},
  publisher = {PMLR}
}

@article{smith2018bayesian, title={Bayesian adaptive bandit-based designs using the Gittins index for multi-armed trials with normally distributed endpoints}, author={Smith, Adam L and Villar, Sofia S}, journal={Journal of Applied Statistics}, volume={45}, number={6}, pages={1052--1076}, year={2018}}

@Article{villar2015multiarmed,
  author    = {Villar, Sofía S. and Bowden, Jack and Wason, James},
  title     = {Multi-Armed Bandit Models for the Optimal Design of Clinical Trials: Benefits and Challenges},
  journal   = {Statistical Science},
  volume    = {30},
  number    = {2},
  pages     = {199--215},
  year      = {2015},
  doi       = {10.1214/14-STS504}
}

@article{williams2021challenges,
  title        = {Challenges in Statistical Analysis of Data Collected by a Bandit Algorithm: An Empirical Exploration in Applications to Adaptively Randomized Experiments},
  author       = {Williams, Joseph J. and Nogas, Jacob and Deliu, Nina and Shaikh, Hammad and Villar, Sofia S. and Durand, Audrey and Rafferty, Anna N.},
  journal      = {arXiv preprint arXiv:2103.12198},
  year         = {2021},
  url          = {https://arxiv.org/pdf/2103.12198}
}

@Article{deliu2021efficient,
  author    = {Deliu, Nina and Williams, Joseph J. and Villar, Sofía S.},
  title     = {Efficient Inference Without Trading-off Regret in Bandits: An Allocation Probability Test for Thompson Sampling},
  journal   = {arXiv preprint},
  volume    = {arXiv:2111.00137},
  year      = {2021},
  url       = {https://arxiv.org/abs/2111.00137}
}

@article{hadad2021confidence,
  author    = {Hadad, Vitor and Hirshberg, David A. and Zhan, Ruohan and Wager, Stefan and Athey, Susan},
  title     = {Confidence Intervals for Policy Evaluation in Adaptive Experiments},
  journal   = {Proceedings of the National Academy of Sciences},
  volume    = {118},
  number    = {15},
  pages     = {e2014602118},
  year      = {2021},
  doi       = {10.1073/pnas.2014602118}
}

@InProceedings{deshpande2018accurate,
  author    = {Deshpande, Yash and Mackey, Lester and Syrgkanis, Vasilis and Taddy, Matt},
  title     = {Accurate Inference for Adaptive Linear Models},
  booktitle = {Proceedings of the 35th International Conference on Machine Learning (ICML) 2018},
year      = {2018},
  publisher = {PMLR}
}

@article{auer2002finite,
  title        = {Finite‐time Analysis of the Multiarmed Bandit Problem},
  author       = {Auer, Peter and Cesa‐Bianchi, Nicol\`o and Fischer, Paul},
  journal      = {Machine Learning},
  volume       = {47},
  number       = {2-3},
  pages        = {235--256},
  year         = {2002},
  doi          = {10.1023/A:1013689704352},
  url          = {https://doi.org/10.1023/A:1013689704352}
}

@inproceedings{yao2021powerconstrained,
  author    = {Yao, Jiayu and Brunskill, Emma and Pan, Weiwei and Murphy, Susan A. and Doshi-Velez, Finale},
  title     = {Power Constrained Bandits},
  booktitle = {Proceedings of the 6th Machine Learning for Healthcare Conference (MLHC)},
  volume    = {149},
year      = {2021},
  publisher = {PMLR}
}

@article{ham2023art,
  title={Hypothesis Testing in Adaptively Sampled Data: {ART} to Maximize Power Beyond iid Sampling},
  author={Ham, Dae Woong and Qiu, Jiaze},
  journal={TEST},
  volume={32},
  number={3},
  pages={998--1037},
  year={2023},
  doi={10.1007/s11749-023-00861-2}
}

@article{simchi-levi2023mabdesign,
  title={Multi-Armed Bandit Experimental Design: Online Decision-Making and Adaptive Inference},
  author={Simchi-Levi, David and Wang, Chonghuan},
  journal={Management Science},
  volume={71},
  number={6},
  pages={4828--4846},
  year={2025},
  doi={10.1287/mnsc.2023.01649}
}

@article{zhong2023achieving,
  title={Achieving the Pareto Frontier of Regret Minimization and Best Arm Identification in Multi-Armed Bandits},
  author={Zhong, Zixin and Cheung, Wang Chi and Tan, Vincent Y. F.},
  journal={Transactions on Machine Learning Research},
  year={2023},
  url={https://openreview.net/forum?id=XXfEmIMJDm}
}

@inproceedings{xiang2022mabvsab,
  title = {Multi-Armed Bandit vs. A/B Tests in E-Commerce: Confidence Interval and Hypothesis Test Power Perspectives},
  author = {Xiang, Ding and West, Rebecca and Wang, Jiaqi and Cui, Xiquan and Huang, Jinzhou},
  booktitle = {Proceedings of the 28th ACM SIGKDD Conference on Knowledge Discovery and Data Mining (KDD 2022)},
year = {2022},
  publisher = {ACM},
  doi = {10.1145/3534678.3539144},
  url = {https://dl.acm.org/doi/10.1145/3534678.3539144}
}

@inproceedings{xu2021unified, title={A unified framework for bandit multiple testing}, author={Xu, Ziyu and Wang, Ruodu and Ramdas, Aaditya}, booktitle={Advances in Neural Information Processing Systems}, volume={34}, pages={16833--16845}, year={2021}}

@book{fisher1925statistical,
  title     = {Statistical Methods for Research Workers},
  author    = {Fisher, Ronald A.},
  year      = {1925},
  publisher = {Oliver and Boyd},
  address   = {Edinburgh, UK},
  edition   = {1st},
}

@article{student1908probable,
  author  = {Student},
  title   = {The probable error of a mean},
  journal = {Biometrika},
  volume  = {6},
  number  = {1},
  pages   = {1--25},
  year    = {1908},
  doi     = {10.1093/biomet/6.1.1},
}

@article{tukey1949comparing,
  author  = {Tukey, John W.},
  title   = {Comparing individual means in the analysis of variance},
  journal = {Biometrics},
  volume  = {5},
  number  = {2},
  pages   = {99--114},
  year    = {1949},
  publisher = {International Biometric Society},
  doi     = {10.2307/3001913},
}

@book{gelman2013bayesian,
  title={Bayesian Data Analysis},
  author={Gelman, Andrew and Carlin, John B and Stern, Hal S and Dunson, David B and Vehtari, Aki and Rubin, Donald B},
  year={2013},
  publisher={CRC Press}
}

@inproceedings{Reza2023ExamEustress,
  author    = {Mohi Reza and Angela Zavaleta Bernuy and Emmy Liu and Tong Li and Zhongyuan Liang and Calista K. Barber and Joseph Jay Williams},
  title     = {Exam Eustress: Designing Brief Online Interventions for Helping Students Identify Positive Aspects of Stress},
  booktitle = {Proceedings of the 2023 CHI Conference on Human Factors in Computing Systems},
  year      = {2023},
doi       = {10.1145/3544548.3581368}
}

@inproceedings{kong2024sequential,
  title     = {Sequential Optimum Test with Multi-armed Bandits for Online Experimentation},
  author    = {Kong, Fang and Zhao, Penglei and Han, Shichao and Wang, Yong and Li, Shuai},
  booktitle = {Proceedings of the 33rd ACM International Conference on Information and Knowledge Management (CIKM)},
  year      = {2024},
publisher = {ACM},
  doi       = {10.1145/3627673.3680040}
}

@book{rosenberger2016randomization,
  title={Randomization in Clinical Trials: Theory and Practice},
  author={Rosenberger, William F. and Lachin, John M.},
  year={2016},
  publisher={Wiley}
}

@article{osband2013more,
  title={(More) efficient reinforcement learning via posterior sampling},
  author={Osband, Ian and Russo, Daniel and Van Roy, Benjamin},
  journal={Advances in Neural Information Processing Systems},
  volume={26},
  year={2013}
}

@InProceedings{bui11committing,
  title     = {Committing Bandits},
  author    = {Bui, Loc and Johari, Ramesh and Mannor, Shie},
  booktitle = {Advances in Neural Information Processing Systems (NeurIPS) 2011},
  year      = {2011},
  url       = {https://papers.nips.cc/paper_files/paper/2011/file/d56b9fc4b0f1be8871f5e1c40c0067e7-Paper.pdf}
}

@Article{kanarios24cabai,
  title     = {Cost Aware Best Arm Identification},
  author    = {Kanarios, Kellen and Zhang, Qining and Ying, Lei},
  journal   = {Reinforcement Learning Journal},
  volume    = {4},
  pages     = {1533--1545},
  year      = {2024}
}

@inproceedings{xia2015budgeted,
  title   = {Thompson Sampling for Budgeted Multi-Armed Bandits},
  author  = {Xia, Yingce and Li, Haifang and Qin, Tao and Yu, Nenghai and Liu, Tie-Yan},
  booktitle = {Proceedings of the 24th International Joint Conference on Artificial Intelligence (IJCAI)},
year    = {2015}
}

@book{zar2010biostatistical,
  title={Biostatistical Analysis},
  author={Zar, Jerrold H},
  year={2010},
  edition={5},
  publisher={Pearson Prentice-Hall},
  address={Upper Saddle River, NJ}
}

\newpage
\appendix

\appendix
\appendix
\appendix
\section{Proof of Theorem~\ref{thm:ait_lrt_optimal}}
\label{app:proof_ait_lrt_optimal}
The proof proceeds by first identifying key structural properties of the
multi-armed bandit (MAB) data-generating process, and then applying the
Neyman--Pearson lemma to the induced distribution over experiment histories.
\paragraph{Properties of the MAB Data-Generating Process.}
We assume the standard MAB process: at each time~$t$, the algorithm~$\pi$
selects an arm $a_t$ based on the observed history~$h_{t-1}$, after which a reward
$r_t$ is drawn from the distribution associated with arm~$a_t$ under~$\vec{\nu}$.
Under this model, the following properties hold.

\begin{property}[Conditional Independence of Rewards]
\label{prop:reward_independence}
Given the selected arm~$a_t$ and the reward distribution specification~$\vec{\nu}$,
the reward~$r_t$ is conditionally independent of the past history~$h_{t-1}$ and the
data-collection algorithm~$\pi$. That is,
\[
p(r_t \mid a_t, h_{t-1}, \vec{\nu}, \pi)
=
p(r_t \mid a_t, \vec{\nu}).
\]
\end{property}

\begin{property}[Conditional Independence of Actions]
\label{prop:action_independence}
Given the past history~$h_{t-1}$ and the data-collection algorithm~$\pi$, the action~$a_t$
is conditionally independent of the reward distribution specification~$\vec{\nu}$. That is,
\[
p(a_t \mid h_{t-1}, \vec{\nu}, \pi)
=
p(a_t \mid h_{t-1}, \pi).
\]
\end{property}

\medskip
\noindent
The next lemma combines these two properties to show that, for simple hypotheses,
the likelihood ratio over complete histories $h_T$ has a particularly simple form.
In particular, it is invariant to the choice of data-collection algorithm $\pi$,
as long as the likelihood ratio is well defined (i.e., $p(h_T \mid \vec{\nu}_i,\pi)>0$
for $i\in\{0,1\}$).

\begin{lemma}[Algorithm Invariance of the Likelihood Ratio]
\label{lem:LR}
Let $\pi$ be any data-collection algorithm such that
$p(h_T \mid \vec{\nu}_0, \pi) > 0$ and $p(h_T \mid \vec{\nu}_1, \pi) > 0$.
Let
\[
h_T = (a_1, r_1, \ldots, a_T, r_T)
\]
denote an ordered experiment history. Then the likelihood ratio satisfies
\[
\frac{p(h_T \mid \vec{\nu}_1, \pi)}{p(h_T \mid \vec{\nu}_0, \pi)}
=
\frac{
\prod_{t=1}^T p(r_t \mid a_t, \vec{\nu}_1)
}{
\prod_{t=1}^T p(r_t \mid a_t, \vec{\nu}_0)
},
\]
and therefore does not depend on the data-collection algorithm~$\pi$. Moreover, the
statistic depends on~$h_T$ only through the collection of rewards observed under each
arm and is invariant to permutations of the time indices.
\end{lemma}

\noindent\emph{Proof.}
By the chain rule of probability,
\[
p(h_T \mid \vec{\nu}, \pi)
=
\prod_{t=1}^T
p(a_t \mid h_{t-1}, \vec{\nu}, \pi)\,
p(r_t \mid a_t, h_{t-1}, \vec{\nu}, \pi).
\]
By Property~\ref{prop:reward_independence} and Property~\ref{prop:action_independence},
this expression simplifies to
\[
p(h_T \mid \vec{\nu}, \pi)
=
\prod_{t=1}^T
p(a_t \mid h_{t-1}, \pi)\,
p(r_t \mid a_t, \vec{\nu}).
\]
Taking the likelihood ratio under $\vec{\nu}_1$ and $\vec{\nu}_0$, the action-selection
terms cancel, yielding
\[
\frac{p(h_T \mid \vec{\nu}_1, \pi)}{p(h_T \mid \vec{\nu}_0, \pi)}
=
\frac{
\prod_{t=1}^T p(r_t \mid a_t, \vec{\nu}_1)
}{
\prod_{t=1}^T p(r_t \mid a_t, \vec{\nu}_0)
}.
\]
\hfill $\square$

\medskip
\noindent
We now complete the proof of Theorem~\ref{thm:ait_lrt_optimal}.
For any fixed algorithm $\pi$, the simple hypotheses $\vec{\nu}_0$ and $\vec{\nu}_1$
induce two fully specified distributions on the sample space $\mathcal{H}_T$, namely
$p(\cdot \mid \vec{\nu}_0,\pi)$ and $p(\cdot \mid \vec{\nu}_1,\pi)$.
By the Neyman--Pearson lemma, the most powerful level-$\alpha$ test for discriminating
between these two distributions rejects for large values of the likelihood ratio
$\frac{p(h_T \mid \vec{\nu}_1,\pi)}{p(h_T \mid \vec{\nu}_0,\pi)}$.

Concretely, define the (classical) LRT rejection region formed from the conditional reward
distributions,
\[
C_k
\;:=\;
\left\{
h_T :
\frac{
\prod_{t=1}^T p(r_t \mid a_t, \vec{\nu}_1)
}{
\prod_{t=1}^T p(r_t \mid a_t, \vec{\nu}_0)
}
> k
\right\}.
\]
Applying AIT amounts to calibrating the threshold $k = k(\alpha,\pi)$ so that
\[
\Pr\!\left( h_T \in C_{k(\alpha,\pi)} \mid \vec{\nu}_0,\pi \right) \le \alpha,
\]
with equality when attainable.
By Lemma~\ref{lem:LR}, this is equivalent to using the rejection region
\[
C
=
\left\{
h_T :
\frac{p(h_T \mid \vec{\nu}_1, \pi)}{p(h_T \mid \vec{\nu}_0, \pi)}
> k(\alpha,\pi)
\right\}.
\]
Therefore, by the Neyman--Pearson lemma, $C$ is the most powerful level-$\alpha$ test for
testing the simple null hypothesis $\vec{\nu}_0$ against the simple alternative $\vec{\nu}_1$
under data collection by $\pi$.\footnote{Although we express the test in terms of the experiment
history $h_T$, this does not affect the application of the Neyman--Pearson lemma. For a fixed
algorithm $\pi$, the induced distribution over $h_T$ is fully specified under both $\vec{\nu}_0$
and $\vec{\nu}_1$, and the likelihood ratio is well defined.}
This proves Theorem~\ref{thm:ait_lrt_optimal}.
\hfill $\square$

\section{More empirical evaluation}\label{subsec:comp_H}

We consider the simulation setting with two-armed binary rewards, and the ground truth of the two arm means are 0.6 and 0.4, with horizon $T=200$. We try to evaluate the power performance of our AIT correction method, and compare it with ART by \cite{ham2023art}. We evaluate it on the $\epsilon$-TS family with $\epsilon$ changing from 0 to 0.8. We consider $T=200$ total samples. 

Our proposed AIT achieves consistently higher power than ART.

\begin{table}[H]
\centering
\begin{tabular}{l c c}
\hline
 & ART  & AIT  \\
\hline
$\epsilon$-TS(0.0) & 0.434 & 0.520 \\
$\epsilon$-TS(0.1) & 0.607 & 0.675 \\
$\epsilon$-TS(0.2) & 0.704 & 0.750 \\
$\epsilon$-TS(0.4) & 0.810 & 0.827 \\
$\epsilon$-TS(0.8) & 0.871 & 0.878 \\
\hline
\end{tabular}
\caption{
Test power comparison under adaptive data collection.
AIT consistently achieves higher power than ART across all algorithms.
}
\label{tab:more_art_ait_comp}
\end{table}


\section{Why naive formulations fail}\label{app:why_naive}
One might wonder why we do not use simpler objectives such as $R - wT$ or $R/T - wT$. 
Besides the fact that they do not satisfy Constraint~1 and distort the meaning of '$w$', we discuss their practical 
problems in more detail.

\textbf{For the objective $R - wT$:}  
Failing Constraint~1 means it can assign a higher score to an experiment that is 
actually worse. For example, consider option A with $(T=100, R=50)$ and choose 
$w=0.2$. Now consider another experiment with $(T=101, R=50.3)$. The second option 
has \emph{both} a lower mean reward and a higher cost, yet $R - wT$ assigns it a 
better score (recall that $R$ is the cumulative, not average, reward).  
In addition, $R - wT$ does not satisfy the intercept-shift property in Constraint~3. 
This makes the choice of $w$ depend heavily on prior knowledge about the average baseline of 
reward. For instance, if $w$ is mistakenly chosen smaller than the average reward across 
arms, then the objective can instantly explode: one can increase $T$ indefinitely, let 
the cumulative reward grow, and obtain arbitrarily large scores.

\textbf{For the objective $R/T - wT$:}  
This formulation under-represents reward and places excessive emphasis on experiment length. Because the average reward $R/T$ exhibits 
sharp diminishing returns as $T$ increases, any reward improvements quickly become 
negligible. As a consequence, this objective overwhelmingly favors algorithms that 
use near-minimal experiment horizons, regardless of reward performance.

\section{Detailed simulation setting for hypothesis tests}\label{app:tests}

We set the type~I error (or family-wise error rate) to $\alpha = 0.05$ and the target power to $1 - \beta = 0.8$.
For power analysis, we assume that each arm has a Bernoulli reward distribution, with arm means
$p \sim D_p = \mathrm{Beta}(5,5)$.
Moreover, for power analysis of all tests except ANOVA, we consider a minimum effect size of $d_0 = 0.1$, and we only compute power when the difference between a tested pair exceeds $d_0$.
We consider the following test setup:

\paragraph{ANOVA test.}
We use the standard one-way ANOVA test~\citep{fisher1925statistical} to compare the mean rewards across arms. 
In our simulation, we apply the test directly on the observed sample means of the three arms.

\paragraph{t-tests.}
A t-test~\citep{student1908probable} assesses whether the means of two groups differ significantly. We consider two variants of the t-test, named \textit{t-constant} and \textit{t-control}. \\

\noindent\textbf{t-constant:} compares each arm’s mean reward to a fixed baseline of $0.5$.
We conduct a \textbf{one-sided} test, counting only cases where an arm’s mean exceeds $0.5$.
For each arm $k$, the null hypothesis is $H_0: p_k = 0.5$, and the alternative hypothesis is $H_1: p_k > 0.5$.
We control the family-wise error rate, while power is averaged across arms.
The family-wise error rate can be expressed as
\[
\Pr(\text{reject } H_0 \text{ for all arms} \mid p_1 = p_2 = p_3 = 0.5),
\]
and the one-sided power (by symmetry across arms) can be written as
\[
\Pr(\text{reject } H_0 \text{ for arm 1} \mid p_1 - 0.5 > d_0,\; p_1,p_2,p_3 \sim D_p).
\]

\noindent\textbf{t-control:} treats the first arm as a control and compares each remaining arm to it.
For each arm $k \neq 1$, the null hypothesis is $H_0: p_k = p_1$.
We conduct a \textbf{two-sided} test with alternative hypothesis $H_1: p_k \neq p_1$.
Unlike t-constant, the null hypothesis location is unrestricted.
Accordingly, the family-wise error rate is
\[
\Pr(\text{reject } H_0 \text{ for all arms} \mid p_1 = p_2 = p_3),
\]
and the two-sided power can be expressed as
\[
\Pr(\text{reject } H_0 \text{ for arm 1 vs.\ arm 2} \mid |p_1 - p_2| > d_0,\; p_1,p_2,p_3 \sim D_p).
\]
We use a two-sided test for t-control to distinguish it from t-constant and to provide a more inclusive evaluation when the direction of effects is not specified in advance.

\paragraph{Tukey tests.} 
A Tukey test~\citep{tukey1949comparing} is a multiple-comparison procedure used to determine which pairs of group means differ significantly while controlling the family-wise error rate across all comparisons. 
Here, we consider testing whether the best arm performs significantly better than each of the other arms. 
Assume that arm $k$ is the best arm. We want to test $H_0: \exists i \neq k,\, p_k \leq p_i$ and $H_1: \forall i \neq k,\, p_k > p_i$. 
We do so by testing each other arm against the best-performing arm, and the result is successful if we identify the best arm and reject the null for each pair. 
The type I error is:
\[
\Pr\!\left(
\text{reject } H_0 \text{ for all arms}
\;\middle|\;
p_1 = p_2 = p_3
\right).
\]
The power is then defined as
\[
\Pr\!\left(
\text{reject } H_0 \text{ for arm } k 
\;\middle|\;
p_k - p_i > d_0,\; \forall i \neq k,\; p_1,p_2,p_3 \sim D_p
\right).
\]

\section{Optimization Interface}
\label{app:user_input}
\begin{figure}[H]
    \centering
    \includegraphics[width=\linewidth]{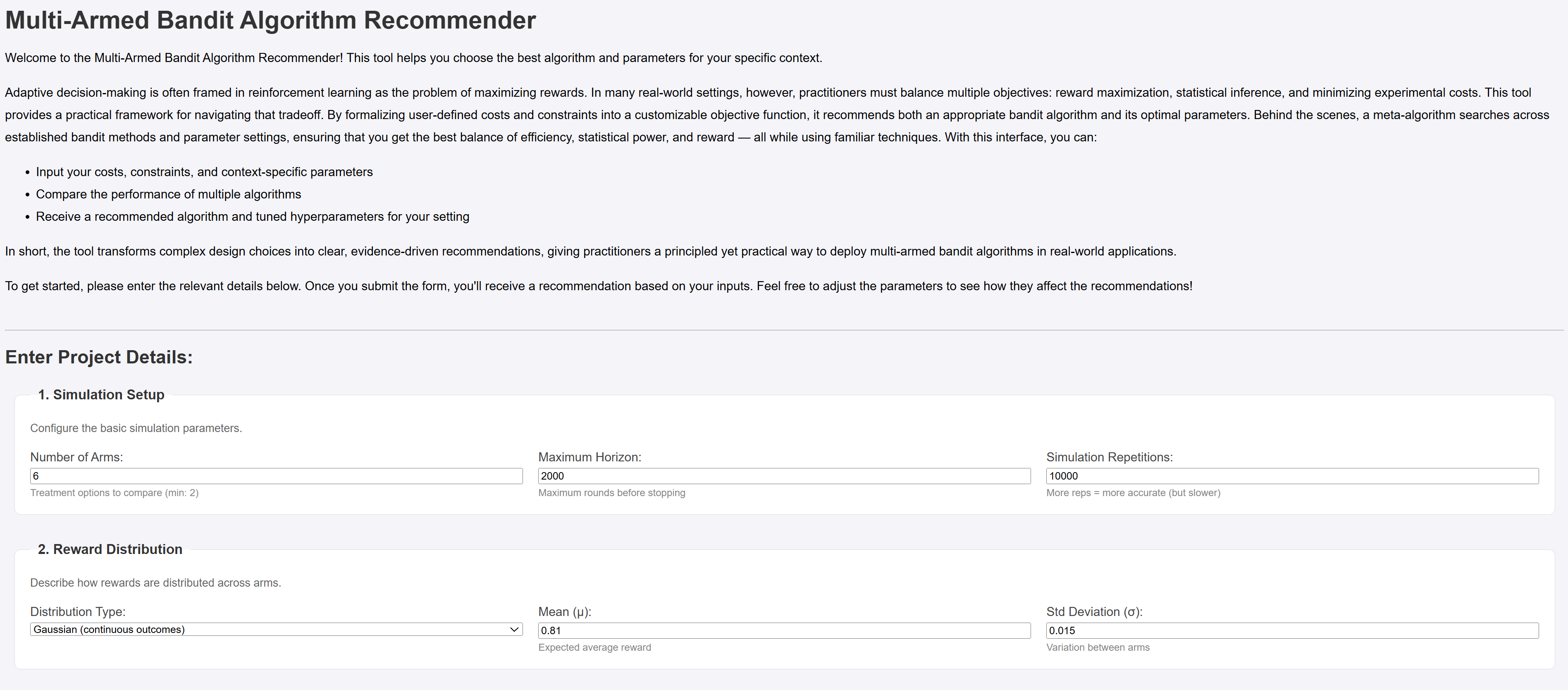}
    \includegraphics[width=\linewidth]{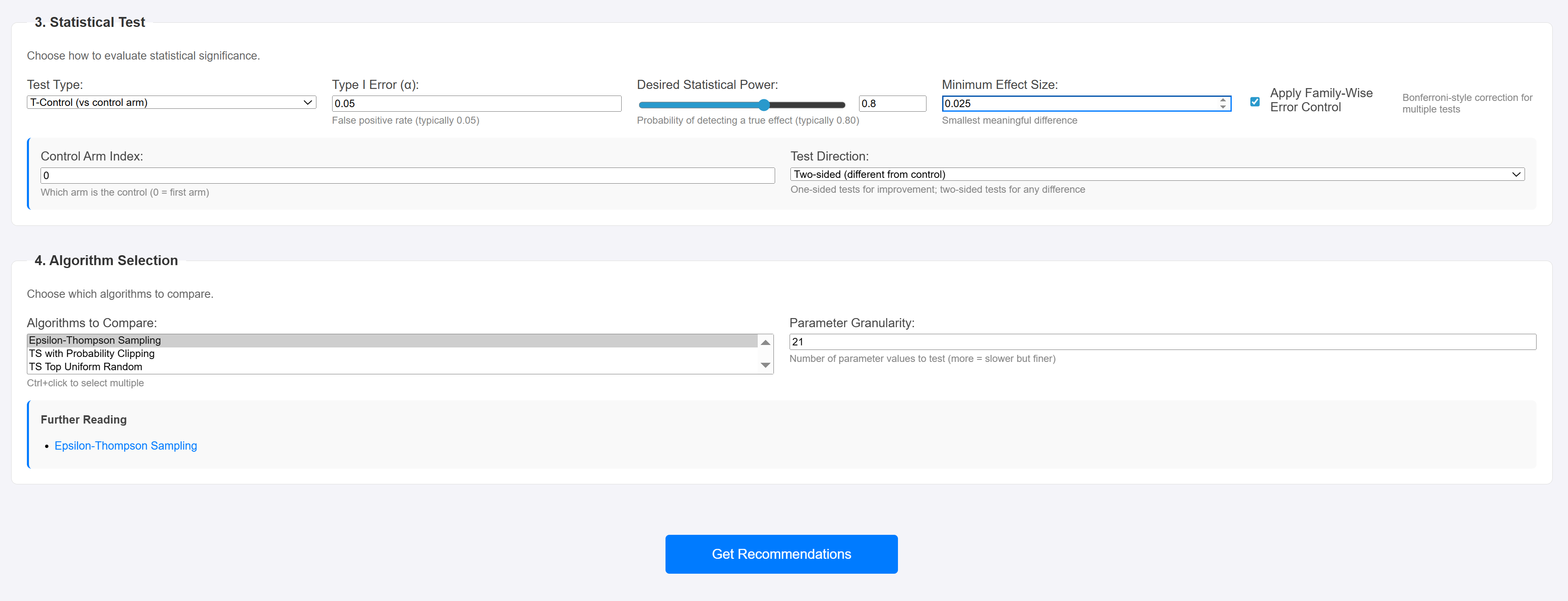}
    \caption{Screenshot of our optimization framework web application user input page.}
    \label{fig:postdiff_ui_input}
\end{figure}

\section{Challenges and efficient methods for MAB experiment simulation}\label{appendix:approximation}
Our main hypothesis testing correction approach (Procedure~\ref{proc:alg_induced_correction}) is simple to implement in a real experimental setting. However, it becomes computationally costly when conducting power analysis simulations for an MAB algorithm.

More specifically, for each replication out of the total \(N\) replications in a power analysis simulation, our test correction (Procedure~\ref{proc:alg_induced_correction}) requires simulating a null-hypothesis scenario and estimating the corresponding critical region threshold. Estimation variance in this threshold not only increases the variance of power estimates, but also introduces systematic bias, dragging estimated power toward \(0.5\). To see this, note that power is computed as $P\!\left(S(h_T) > \hat{s}\right)$,  where \(S(h_T)\) is the test statistic computed from experimental outcome \(h_T\), and \(\hat{s}\) is the estimated critical region threshold. If the variance of \(\hat{s}\) becomes large, the rejection probability converges to \(0.5\), regardless of the true alternative. Empirically, we find that at least \(M = 500\) null-hypothesis replications are typically required to reduce power estimation bias to within \(1\%\). This leads to a total computational cost of \(M \times N\) bandit runs, which is substantially larger than a standard MAB simulation that requires only \(N\) replications.

To address this challenge, we combine three complementary techniques that substantially reduce both runtime and memory usage. First, we fully vectorize the implementation and execute simulation replications jointly. Second, we reduce the total number of null-distribution simulations by binning similar estimated null hypotheses and simulating only a constant number of representative grid points. Critical region thresholds for intermediate null specifications are then obtained via interpolation (see Procedure~\ref{proc:approximation}). Lastly, we accelerate each bandit replication and reduce memory usage by updating the bandit algorithm in batches rather than at every step, and by aggregating results within each batch. The resulting batched adaptive experiment process is shown in Procedure~\ref{proc:batched_bandit}, which serves as an approximation to the standard adaptive experiment process in Procedure~\ref{proc:bandit}. Below we provide additional explanation for Procedure~\ref{proc:batched_bandit}..

Concretely, at time \(t\), the total batch size grows geometrically, while the number of distinct actions sampled within that batch scales at a cube-root rate, i.e., proportional to \(t^{1/3}\). Each selected action is then evaluated using multiple independent reward draws. Under this design, the effective influence of newly arriving data remains controlled at a fixed proportion relative to the existing data, while the allocation distribution across arms remains sufficiently granular. Empirically, for a two-arm Thompson Sampling experiment with horizon \(T=200\), Bernoulli rewards, and means \((\mu_1,\mu_2)=(0.6,0.4)\), the approximation error in average reward is below $0.001$.

Meanwhile, for each batched action, we store both the aggregated reward sum \(R = \sum_j r_j\) and the aggregated squared reward sum \(R^2 = \sum_j r_j^2\). In standard MAB settings with Bernoulli rewards, storing \(R\) alone is sufficient, since \(r^2 = r\). However, our power analysis framework is designed to support more general reward models. For example, when updating a Bayesian model with Gaussian rewards, both the sample mean and variance are required. If rewards were stored only via aggregated sums \(R\), the sample variance could not be recovered. Storing \(R^2\) preserves second-moment information while maintaining a compact, batch-level representation of the data.

Combining these three techniques, our simulations achieve orders-of-magnitude speedups compared to the naïve implementation. For example, under the naïve design, a single configuration with \(N = 10{,}000\) replications (assuming an additional \(M = 500\) null simulations without approximation) and horizon \(T = 200\) can take over an hour to run and require roughly \(20\)~GB of memory, whereas our proposed approach completes in under a few seconds with negligible memory usage.

\section{Procedures}\label{appendix:algorithms}

\begin{procedure}[H]
\caption{Regular Adaptive Experiment Process}\label{proc:bandit}
\begin{algorithmic}[1]
\State \textbf{Input:} Number of arms $K$; arm reward distributions $\vec{\nu}$; experiment horizon $T$; algorithm $\pi$
\Procedure{Bandit}{$K,\, \vec{\nu},\, T,\, \pi$}
    \State Initialize empty history $H_0 = ()$
    \For{$t = 1$ to $T$}
        \State Select arm $a_t \sim \pi(h_{t-1})$
        \State Observe reward $r_t \sim \nu_{a_t}$
        \State Update history $h_t = (a_1, r_1, \ldots, a_t, r_t)$
    \EndFor
\EndProcedure
\State \textbf{Output:} Experiment history $h_T = (a_1, r_1, \ldots, a_T, r_T)$
\end{algorithmic}
\end{procedure}

\begin{procedure}[H]
\caption{Approximation Procedure for Adaptive Experimentation Power Analysis}\label{proc:approximation}
\begin{algorithmic}[1]
\State \textbf{Input:} Number of arms $K$; power analysis prior $D_{\vec{\nu}}$; experiment horizon $T$; algorithm $\pi$; test statistic $S$, type I error constraint $\alpha_0$; number of simulation repetitions $M$, number of approximation grid points $B$.
\Procedure{PowerAnalysis}{$K,\, D_{\vec{\nu}}, \, T,\,  \pi,\, S,\, \alpha_0,\, M, \, B$}
    \For{$m = 1$ to $M$}
        \State Sample $\vec{\nu} \sim D_{\vec{\nu}}$
        \State Obtain $h_{m,T}$ by running \textbf{Bandit}$(K, \vec{\nu}, T, \pi)$ 
        \State Record cumulative mean reward $\bar{r}_{m,t}$ from $h_{m,t}$ for $t=1,...,T$
        \State Get MLE of the $H_0$ parameter $\hat{\theta}_{H_0,m}$ using $h_{m,T}$
    \EndFor
    \State Select $B$ equally spaced grid points between $\min(\hat{\theta}_{H_0,m})$ and $\max(\hat{\theta}_{H_0,m})$: $\hat{\theta}'_{H_0,1}, \ldots, \hat{\theta}'_{H_0,B}$
    \For{$b = 1$ to $B$}
         \State Set $\nu_{0,b}$ to be the distribution under parameterization $\hat{\theta}'_{H_0,b}$
        \State Perform \textbf{AlgoInducedCrit}$(K,T,\nu_{0,b}, S,\pi,\alpha_0,M/B)$ to obtain adjusted critical region sequence $\{C'_{b,t}\}_{t=1}^T$
    \EndFor
    \For{$m = 1$ to $M$}
        \For{$t=1$ to $T$}
            \State Get interpolated $C_{m,t}$ from $\{C'_{b,t}\}_{b=1}^B$ based on $\hat{\theta}_{H_0,m}$
            \State Compute $S(h_{m,t})$ and perform hypothesis testing
            \State Record test result $x_{m,t} = \mathbb{1}_{ S(h_{m,t}) \in C_{m,t} }$
    
        \EndFor
    \EndFor
\EndProcedure
\State \textbf{Output:} Type II error vector $\vec{\beta}=(\hat{\beta}_1,...,\hat{\beta}_T)$ where
$\hat{\beta}_t = 1 - \frac{1}{M}\sum_{m=1}^M x_{m,t}$; cumulative mean reward vector $\vec{r}=(\bar{r}_1,...,\bar{r}_T)$ where $\bar{r}_t = \frac{1}{M}\sum_{m=1}^M \bar{r}_{m,t}$
\end{algorithmic}
\end{procedure}

\begin{procedure}[H]
\caption{Batched Adaptive Experiment Process}
\label{proc:batched_bandit}
\begin{algorithmic}[1]
\State \textbf{Input:} Number of arms $K$; arm reward distributions $\vec{\nu}$; experiment horizon $T$; algorithm $\pi$
\State Initialize compressed history $h_0 = ()$
\State $t \gets 0$

\While{$t < T$}

    \State $s_t \gets \mathrm{round}(1 + 0.05\, t)$          \Comment{step size}
    \State $n_t \gets \mathrm{round}(s_t^{1/3})$            \Comment{number of batched actions}
    \State $m \gets s_t / n_t$                              \Comment{repetitions per action}

    \For{$i = 1$ to $n_t$}
        \State Sample action $a_{t,i} \sim \pi(h_t)$
        \State Observe $m$ rewards $r_{t,i,1:m} \sim \nu_{a_{t,i}}$

        \State Compute compressed entry:
        \[
           R_{t,i} = \sum_{j=1}^m r_{t,i,j},
           \qquad
           R_{t,i}^2 = \sum_{j=1}^m r_{t,i,j}^2
        \]
    \EndFor

    \State Append history:
    \[
        h_t \gets h_t \cup
        \big(a_{t,1}, R_{t,1}, R_{t,1}^2,\;
              \dots,\;
              a_{t,n_t}, R_{t,n_t}, R_{t,n_t}^2\big)
    \]

    \State $t \gets t + s_t \cdot n_t$
\EndWhile

\State \textbf{Output:} Compressed history $h_T$
\end{algorithmic}
\end{procedure}

\begin{procedure}[h] 
\caption{Objective Function Optimization Procedure}\label{proc:optimization}
\begin{algorithmic}[1]

\State \textbf{Input:} Number of arms $K$; power analysis prior $D_{\vec{\nu}}$; maximum acceptable experiment horizon $T_{\max}$; algorithm $\pi$; list of $n$ algorithm parameters $\vec{\phi}=(\phi_1,\dots,\phi_n)$ to be optimized from; test statistic $S$; type I and type II error constraints $\alpha_0$, $\beta_0$; objective score function $F$ with weight $w$; number of simulation repetitions $M$; number of approximation grid points $B$.

\Procedure{ObjOpt}{$K,\, D_{\vec{\nu}},\, T_{\max},\, \pi,\, \vec{\phi},\, S,\, \alpha_0,\, \beta_0,\, F,\, w,\, M,\, B$}

\State Initialize admissible parameter set $\Phi=\{\}$

\For{$i = 1,\ldots,n$}

    \State Run
    \textbf{PowerAnalysis}$(K, D_{\vec{\nu}}, T_{\max}, \pi_{\phi_i}, S, \alpha_0, M, B)$ and obtain $\vec{\beta}_i$ and $\vec{r}_i$

    \If{\(\exists\, t \in [T_{\max}] \text{ such that } \hat{\beta}_{i,t} \le \beta_0\)}
        \State Record $T_{\phi_i} = \min\{\, t : \hat{\beta}_{i,t} \leq \beta_0 \,\}$
        \State Record $r_{\phi_i} = r_{i, T_{\phi_i}}$ from $\vec{r}_i$
        \State Add $\phi_i$ to $\Phi$
    \Else
        \State \textbf{break}
    \EndIf

\EndFor

\State $\phi^* \gets \arg\min_{\phi_i \in \Phi} F(T_{\phi_i}, r_{\phi_i}, w)$
\EndProcedure
\State \textbf{Output:} suggested parameter $\phi^*$
\end{algorithmic}
\end{procedure}

\end{document}